\documentclass[11pt]{article}

\usepackage[utf8]{inputenc}
\usepackage[T1]{fontenc}
\usepackage{lmodern}
\usepackage[DIV=11]{typearea} 

\usepackage{microtype}

\usepackage{amsmath}
\usepackage{amssymb}
\usepackage{amsthm}
\usepackage{thmtools}

\usepackage{xcolor}
\usepackage{xspace}

\usepackage[
	style=alphabetic,
	backref=true,
	doi=false,
	url=false,
	maxcitenames=3,
	mincitenames=3,
	maxbibnames=10,
	minbibnames=10,
	backend=bibtex8,
	sortlocale=en_US
]{biblatex}

\usepackage[ocgcolorlinks]{hyperref} 
\usepackage{cleveref}

\makeatletter
\g@addto@macro\bfseries{\boldmath}
\makeatother

\makeatletter
\g@addto@macro\mdseries{\unboldmath}
\g@addto@macro\normalfont{\unboldmath}
\g@addto@macro\rmfamily{\unboldmath}
\g@addto@macro\upshape{\unboldmath}
\makeatother

\DeclareCiteCommand{\citem}
    {}
    {\mkbibbrackets{\bibhyperref{\usebibmacro{postnote}}}}
    {\multicitedelim}
    {}

\renewcommand*{\multicitedelim}{\addcomma\space}

\AtEveryCitekey{\clearfield{note}}

\newcommand{\myhref}[1]{%
  \iffieldundef{doi}
    {\iffieldundef{url}
       {#1}
       {\href{\strfield{url}}{#1}}}
    {\href{http://dx.doi.org/\strfield{doi}}{#1}}%
}

\DeclareFieldFormat{title}{\myhref{\mkbibemph{#1}}}
\DeclareFieldFormat
  [article,inbook,incollection,inproceedings,patent,thesis,unpublished]
  {title}{\myhref{\mkbibquote{#1\isdot}}}

\makeatletter
\AtBeginDocument{%
    \newlength{\temp@x}%
    \newlength{\temp@y}%
    \newlength{\temp@w}%
    \newlength{\temp@h}%
    \def\my@coords#1#2#3#4{%
      \setlength{\temp@x}{#1}%
      \setlength{\temp@y}{#2}%
      \setlength{\temp@w}{#3}%
      \setlength{\temp@h}{#4}%
      \adjustlengths{}%
      \my@pdfliteral{\strip@pt\temp@x\space\strip@pt\temp@y\space\strip@pt\temp@w\space\strip@pt\temp@h\space re}}%
    \ifpdf
      \typeout{In PDF mode}%
      \def\my@pdfliteral#1{\pdfliteral page{#1}}
      \def\adjustlengths{}%
    \fi
    \ifxetex
      \def\my@pdfliteral #1{}
      \def\adjustlengths{\setlength{\temp@h}{-\temp@h}\addtolength{\temp@y}{1in}\addtolength{\temp@x}{-1in}}%
    \fi%
    \def\Hy@colorlink#1{%
      \begingroup
        \ifHy@ocgcolorlinks
          \def\Hy@ocgcolor{#1}%
          \my@pdfliteral{q}%
          \my@pdfliteral{7 Tr}
        \else
          \HyColor@UseColor#1%
        \fi
    }%
    \def\Hy@endcolorlink{%
      \ifHy@ocgcolorlinks%
        \my@pdfliteral{/OC/OCPrint BDC}%
        \my@coords{0pt}{0pt}{\pdfpagewidth}{\pdfpageheight}%
        \my@pdfliteral{F}
        \my@pdfliteral{EMC/OC/OCView BDC}%
        \begingroup%
          \expandafter\HyColor@UseColor\Hy@ocgcolor%
          \my@coords{0pt}{0pt}{\pdfpagewidth}{\pdfpageheight}%
          \my@pdfliteral{F}
        \endgroup%
        \my@pdfliteral{EMC}%
        \my@pdfliteral{0 Tr}
        \my@pdfliteral{Q}%
      \fi
      \endgroup
    }%
}
\makeatother

\colorlet{DarkRed}{red!50!black}
\colorlet{DarkGreen}{green!50!black}
\colorlet{DarkBlue}{blue!50!black}

\hypersetup{
	linkcolor = DarkRed,
	citecolor = DarkGreen,
	urlcolor = DarkBlue,
	bookmarks = true,
	bookmarksnumbered = true,
	linktocpage = true
}

\declaretheorem[numberwithin=section]{theorem}
\declaretheorem[numberlike=theorem]{lemma}

\declaretheorem[numberlike=theorem]{corollary}

\newcommand{\inc}{\mathit{In}}
\newcommand{\out}{\mathit{Out}}

\newcommand{\dist}{d}
\newcommand{\inputdegree}{s}
\newcommand{\outputdegree}{d}
\newcommand{\eps}{\epsilon}

\title{Fully Dynamic Spanners with Worst-Case Update Time\thanks{To be presented at the European Symposium on Algorithms (ESA) 2016. This work was partially done while the authors were visiting the Simons Institute for the Theory of Computing.}} 
\author{
Greg Bodwin\thanks{Stanford University}
\and
Sebastian Krinninger\thanks{Max Planck Institute for Informatics}
}

\date{}

\hypersetup{
	pdftitle = {Fully Dynamic Spanners with Worst-Case Update Time},
	pdfauthor = {Greg Bodwin, Sebastian Krinninger}
}

\begin{document}
\maketitle
\begin{abstract}
An {\em $\alpha$-spanner} of a graph $ G $ is a subgraph $ H $ such that $ H $ preserves all distances of $ G $ within a factor of $ \alpha $.
In this paper, we give fully dynamic algorithms for maintaining a spanner $ H $ of a graph $ G $ undergoing edge insertions and deletions with worst-case guarantees on the running time after each update.
In particular, our algorithms maintain:
\begin{itemize}
\item a $3$-spanner with $ \tilde O (n^{1+1/2}) $ edges
with worst-case update time $ \tilde O (n^{3/4}) $, or
\item a $5$-spanner with $ \tilde O (n^{1+1/3}) $ edges with worst-case update time $ \tilde O (n^{5/9}) $.
\end{itemize}
These size/stretch tradeoffs are best possible (up to logarithmic factors).
They can be extended to the weighted setting at very minor cost.
Our algorithms are randomized and correct with high probability against an oblivious adversary.
We also further extend our techniques to construct a $5$-spanner with suboptimal size/stretch tradeoff, but improved worst-case update time.

To the best of our knowledge, these are the {\em first} dynamic spanner algorithms with sublinear worst-case update time guarantees.
Since it is known how to maintain a spanner using small {\em amortized} but large {\em worst-case} update time \citem[Baswana et al.\ SODA'08]{BaswanaKS12}, obtaining algorithms with strong worst-case bounds, as presented in this paper, seems to be the next natural step for this problem.

\end{abstract}
\newpage

\tableofcontents
\newpage

\section{Introduction}

An {\em $\alpha$-spanner} of a graph $G$ is a sparse subgraph that preserves all original distances within a multiplicative factor of $\alpha$.
Spanners are an extremely important and well-studied primitive in graph algorithms.
They were formally introduced by Peleg and Sch{\" a}fer \cite{PS89} in the late eighties after appearing naturally in several network problems \cite{PU89}.
Today, they have been successfully applied in diverse fields such as routing schemes \cite{Cowen01, CW04, PU89, RTZ08, TZ01}, approximate shortest paths algorithms \cite{DHZ96, Elkin05, BaswanaK10}, distance oracles \cite{BaswanaK10, Chechik14, Chechik15, PR10, TZ05}, broadcasting \cite{FPZ+04}, etc.
A landmark upper bound result due to Awerbuch \cite{Awerbuch85} states that for any integer $ k $, every graph has a $(2k-1)$-spanner on $O(n^{1 + 1/k})$ edges.
Moreover, the extremely popular {\em girth conjecture} of Erd\H{o}s \cite{girth} implies the existence of graphs for which $\Omega(n^{1 + 1/k})$ edges are necessary in any $(2k-1)$-spanner.
Thus, the primary question of the optimal sparsity of a graph spanner is essentially resolved.

The next natural question in the field of spanners is to obtain efficient algorithms for computing a sparse spanner of an input graph $G$.
This problem is well understood in the static setting; notable results include \cite{Awerbuch85, BaswanaS07, RTZ05, TZ01}.
However, in many of the above applications of spanners, the underlying graph can experience minor changes and the application requires the algorithm designer to have a spanner available at all times.
Here, it is very wasteful to recompute a spanner from scratch after every modification.  
The challenge is instead to {\em dynamically maintain} a spanner under edge insertions and deletions with only a small amount of time required per update.
This is precisely the problem we address in this paper.

The pioneering work on dynamic spanners was by Ausiello et al.~\cite{AusielloFI06}, who showed how to maintain a $3$- or $5$-spanner with amortized update time proportional to the maximum degree~$ \Delta $ of the graph, i.e. for any sequence of $u$ updates the algorithm takes time $O (u \cdot \Delta)$ in total.
In sufficiently dense graphs, $ \Delta $ might be $\Omega(n)$.
Elkin \cite{Elkin11} showed how to maintain a $(2k-1)$ spanner of optimal size using $\tilde O (mn^{-1/k}) $ expected update time; i.e. {\em super-linear} time for dense enough graphs.
Finally, Baswana et al.~\cite{BaswanaKS12} gave fully dynamic algorithms that maintain $(2k-1)$-spanners with essentially optimal size/stretch tradeoff using {\em amortized} $O(k^2 \log^2 n)$ or $O(1)^k$ time per update.
Their {\em worst-case} guarantees are much weaker: any individual update in their algorithm can require $\Omega(n)$ time.
It is very notable that {\em every} previously known fully dynamic spanner algorithm carries the drawback of $\Omega(n)$ worst-case update time.
It is thus an important open question whether this update time is an intrinsic part of the dynamic spanner problem, or whether this linear time threshold can be broken with new algorithmic ideas.

There are concrete reasons to prefer worst-case update time bounds to their amortized counterparts.
In real-time systems, hard guarantees on update times are often needed to serve each request before the next one arrives.
Amortized guarantees, meanwhile, can cause undesirable behavior in which the system periodically stalls on certain inputs.
Despite this motivation, good worst-case update times often pose a veritable challenge to dynamic algorithm designers, and are thus significantly rarer in the literature.
Historically, the fastest dynamic algorithms usually first come with amortized time bounds, and comparable worst-case bounds are achieved only after considerable research effort.
For example, this was the case for the dynamic connectivity problem on undirected graphs~\cite{KapronKM13} and the dynamic transitive closure problem on directed graphs~\cite{Sankowski04}.
In other problems, a substantial gap between amortized and worst-case algorithms remains, despite decades of research.
This holds in the cases of fully dynamically maintaining minimum spanning trees~\cite{HolmLT01,Frederickson85,EppsteinGIN97}, all-pairs shortest paths~\cite{DemetrescuI04,Thorup05}, and more.
Thus, strong amortized update time bounds for a problem do not at all imply the existence of strong worst-case update time bounds, and once strong amortized algorithms are found it becomes an important open problem to discover whether or not there are interesting worst-case bounds to follow.

The main result of this paper is that highly nontrivial worst-case time bounds are indeed available for fully dynamic spanners.  We present the first ever algorithms that maintain spanners with essentially optimal size/stretch tradeoff and {\em polynomially sublinear} (in the number of nodes in the graph) worst-case update time.  Our main technique is a very general new framework for boosting the performance of an orientation-based algorithm, which we hope can have applications in related dynamic problems.

\subsection{Our results}

We obtain fully dynamic algorithms for maintaining spanners of graphs undergoing edge insertions and deletions.
In particular, in the unweighted setting we can maintain:
\begin{itemize}
	\item a $3$-spanner of size $ O (n^{1+1/2} \log^{1/2}{n} \log{\log{n}}) $ with worst-case update time $ O (n^{3/4} \log^{4}{n}) $, or
	\item a $5$-spanner of size $ O (n^{1+1/3} \log^{2/3}{n} \log{\log{n}}) $ with worst-case update time $ O (n^{5/9} \log^{4}{n}) $, or
	\item a $5$-spanner of size $ O (n^{1+1/2} \log^{1/2}{n} \log{\log{n}}) $ with worst-case update time $ O (n^{1/2} \log^{4}{n}) $.
\end{itemize}
Naturally, these results assume that the initial graph is empty; otherwise, a lengthy initialization step is unavoidable.

Using standard techniques, these results can be extended into the setting of arbitrary positive edge weights, at the cost of an increase in the stretch by a factor of $1 + \eps$ and an increase in the size by a factor of $\log_{1 + \eps} W$ (for any $\eps > 0$, where $W$ is the ratio between the largest and smallest edge weights).

Our algorithms are randomized and correct with high probability against an \emph{oblivious adversary}~\cite{Ben-DavidBKTW94} who chooses its sequence of updates independently from the random choices made by the algorithm.\footnote{In particular, this means that the adversary is \emph{not} allowed to see the current edges of the spanner.}
This adversarial model is the same one used in the previous randomized algorithms with amortized update time~\cite{BaswanaKS12}.
Since the girth conjecture has been proven unconditionally for $ k=2 $ and $ k=3 $ \cite{Wenger91}, the first two spanners have optimal size/stretch tradeoff (up to the $\log$ factor).
The third result sacrifices a non-optimal size/stretch tradeoff in exchange for improved update time.

\subsection{Technical Contributions}

Our main new idea is a general technique for boosting the performance of orientation-based algorithms.

Our algorithm contains three new high-level ideas.
First, let $\vec{G}$ be an arbitrary orientation of the input graph $ G $; i.e. replace every undirected edge $ \{ u, v \} $ by a directed edge, either $ (u, v) $ or $ (v, u) $.
We give an algorithm ALG for maintaining either a $3$-spanner or a $5$-spanner of $ G $ with update time proportional to the maximum \emph{out-degree} of the oriented graph $ \vec{G} $.
This algorithm is based on the clustering approach used in \cite{BaswanaS07}.
For maintaining $3$- and $5$-spanners we only need to consider clusters of diameter at most $ 2 $ consisting of the set of neighbors of certain cluster centers.

This alone is of course not enough, as generally the maximum out-degree of $ \vec{G} $ can be as large as $ n-1 $.
To solve this problem, we combine ALG with the following simple out-degree reduction technique.
Partition outgoing edges of every node into at most $ t \leq \lceil n / \inputdegree \rceil $ groups of size at most $ \inputdegree $ each.
For any $ 1 \leq i \leq t $, we combine the edges of the $i$-th groups and on the corresponding subgraph $ G_i $ we run an instance of ALG to maintain a $3$-spanner with update time $ O (s) $, the maximum out-degree in $ \vec{G}_i $.
By the decomposability of spanners, the union of all these sub-spanners $ H_1 \cup \dots H_t $ is a $3$-spanner of $ G $.
In this way we can obtain an algorithm for maintaining a $3$-spanner of size $ |H_1| + \dots |H_t| = O (n^{5/2} / \inputdegree) $ with worst-case update time $ O (\inputdegree) $ for any $ 1 \leq \inputdegree \leq n $.
We remark that the general technique of partitioning a graph into subgraphs of low out-degree has been used before, e.g. \cite{BE13}; however, our recursive conversion of these subgraphs into spanners is original and an important technical contribution of this paper.

The partitioning is still not enough, as the optimal size of a $3$-spanner is $O(n^{3/2})$, which would then require $s = \Omega(n)$ worst-case update time.
However, we can improve upon this tradeoff once more with a more fine-grained application of ALG.
In particular, on each subgraph $ \vec{G}_i $, ALG maintains two subgraphs $ A_i^1 $ and $ \vec{B}_i^1 $, such that:
\begin{itemize}
\item $ A_i^1 $ is a `partial' $3$-spanner of $ G_i $ of size $ \tilde O (n^{1 + 1/2} \cdot s / n) $, and
\item The maximum out-degree in $ \vec{B}_i^1 $ is considerably smaller than the maximum out-degree in $ \vec{G}_i $.
\end{itemize}
We then recursively apply ALG on $ \vec{B}_1^1 \cup \dots \cup \vec{B}_t^1 $ to some depth $ \ell $ at which the out-degree can no longer be reduced by a meaningful amount.
Our final spanner is then the union of all the sets $A_i^j$, for $1 \le i \le t$ and $1 \le j \le \ell$, as well as the ``remainder'' graphs $\vec{B}_1^\ell \cup \dots \cup \vec{B}_t^\ell$, which have low out-degree and are thus sparse.

In principle, the recursive application of ALG could be problematic, as one update in~$ G $ could lead to several changes to the edges in the $ B_i^1 $ subgraphs, which then propagate as an increasing number of updates in the recursive calls of the algorithm.
This places another constraint on ALG.
We carefully design ALG in such a way that it performs only a constant number of changes to each $ B_i^1 $ with any update in $ G $, and we only recurse to depth $\ell = o(\log n)$ so that the total number of changes at each level is subpolynomial.

Overall, we remark that our framework for performing out-degree reduction is fairly generic, and seems likely applicable to other algorithms that admit the design of an ALG with suitable properties.
The main technical challenges are designing ALG with these properties, and performing some fairly involved parameter balancing to optimize the running time used by the recursive calls.
However, we do not know how to extend our approach to sparser spanners with larger stretches since corresponding constructions usually need clusters of larger diameter and maintaining such clusters with update time proportional to the maximum (out)-degree of the graph seems challenging.

\subsection{Other Related Work}

There has been some related work attacking the spanner problem in other models of computation.
Some of the work on streaming spanner algorithms, in particular \cite{Baswana08,FeigenbaumKMSZ08}, was converted into purely \emph{incremental} dynamic algorithms, which maintain spanners under edge insertions but cannot handle deletions.
This line of research culminated in an incremental algorithm with worst-case update time $ O (1) $ per edge insertion~\cite{Elkin11}.
Elkin \cite{Elkin07} also gave a near-optimal algorithm for maintaining spanners in the distributed setting.

A concept closely related to spanners are \emph{emulators} \cite{DHZ96}, in which the graph $ H $ for approximately preserving the distances may contain arbitrary weighted edges and is not necessarily a subgraph of $ G $.
Dynamic algorithms for maintaining emulators have been commonly used as subroutines to obtain faster dynamic algorithms for maintaining (approximate) shortest paths or distances.
Some of the work on this problem includes \cite{RodittyZ12, BernsteinR11, HenzingerKNFOCS14, HenzingerKNSODA14, ACT14, AbrahamC13}.

As outlined above, one of the main technical contributions of this paper is a framework for exploiting orientations of undirected graphs.
The idea of orienting undirected graphs has been key to many recent advances in dynamic graph algorithms.
Examples include~\cite{NeimanS13,KopelowitzKPS14,PelegS16,AmirKLPPS15,AbrahamDKKP16}.

\section{Preliminaries}

We consider unweighted, undirected graphs $ G = (V, E) $ undergoing edge insertions and edge deletions.
For all pairs of nodes $ u $ and $ v $ we denote by $ \dist_G (u, v) $ the distance between $ u $ and $ v $ in $ G $.
An \emph{$\alpha$-spanner} of a graph $G = (V, E)$ is a subgraph $H = (V, E') \subseteq G$ such that $\dist_H(u, v) \le \alpha \cdot \dist_G(u, v)$ for all $u, v \in V$.\footnote{If $u$ and $v$ are disconnected in $G$, then $\dist_G(u, v) = \infty$ and so they may be disconnected in the spanner as well.}
The parameter $\alpha$ is called the \emph{stretch} of the spanner.
We will use the well-known fact that it suffices to only span distances over the edges of $G$.
\begin{lemma} [Spanner Adjacency Lemma (Folklore)] \label{lem:span adjacency}
If $H = (V, E')$ is a subgraph of $G = (V, E)$ that satisfies $\dist_H(u, v) \le \alpha \cdot \dist_G(u, v)$ for all $(u, v) \in E$, then $H$ is an $\alpha$-spanner of $G$.
\end{lemma}

We will work with {\em orientations} of undirected graphs.
We denote an undirected edge with endpoints $ u $ and $ v $ by $ \{u, v\} $ and a directed edge from $ u $ to $ v $ by $ (u, v) $.
An \emph{orientation} $ \vec{G} = (V, \vec{E}) $ of an undirected graph $ G = (V, E) $ is a directed graph on the same set of nodes such that for every edge $ \{u, v\} $ of $ G $, $ \vec{G} $ either contains the edge $ (u, v) $ or the edge $ (v, u) $.
Conversely, $ G $ is the \emph{undirected projection} of $ \vec{G} $.
In an undirected graph $ G $, we denote by $ N (v) := \{ w \mid \{v, w\} \in G \} $ the set of neighbors of $ v $.
In an oriented graph $ \vec{G} $, we denote by $ \out (v) := \{ w \mid (v, w) \in \vec{G} \} $ the set of outgoing neighbors of $ v $.
Similarly, by $ \inc (v) := \{ u \mid (u, v) \in \vec{G} \} $ we denote the set of incoming neighbors of $ v $.
We denote by $ \Delta^+ (\vec{G}) $ the maximum out-degree of $ \vec{G} $.

Our algorithms can easily be extended to graphs with edge weights, via the standard technique of weight binning:
\begin{lemma}[Weight Binning, e.g.~\cite{BaswanaKS12}]
Suppose there is an algorithm that dynamically maintains a spanner of an arbitrary unweighted graph with some particular size, stretch, and update time.
Then for any $\eps > 0$, there is an algorithm that dynamically maintains a spanner of an arbitrary graph with positive edge weights, at the cost of an increase in the stretch by a factor of $ 1 + \eps $ and an increase in the update time by a factor of $ O(\log_{1 + \eps} W)$ (and no change in update time).  Here, $W$ is the ratio between the largest and smallest edge weight in the graph.
\end{lemma}
Since this extension is already well known, we will not discuss it further.
Instead, we will simplify the rest of the paper by focusing only on the unweighted setting; that is, all further graphs in this paper are unweighted and undirected.

In our algorithms, we will use the well-known fact that good hitting sets can be obtained by random sampling.
This technique was first used in the context of shortest paths by Ullman and Yannakakis~\cite{UllmanY91}.
A general lemma on the size of the hitting set can be formulated as follows.

\begin{lemma}[Hitting Sets] \label{lem:random hitting set}
Let $ a \geq 1 $, let $ V $ be a set of size $ n $ and let $ U_1, U_2, \ldots, U_r $, be subsets of $ V $ of size at least $ q $.
Let~$ S $ be a subset of $ V $ obtained by choosing each element of $ V $ independently at random with probability $ p = \min (x / q, 1) $ where $ x = a \ln{(r n)} + 1 $.
Then, with high probability (whp), i.e. probability at least $ 1 - 1/n^a $, both the following two properties hold:
\begin{enumerate}
\item For every $ 1 \leq i \leq r $, the set $ S $ contains a node in $ U_i $, i.e. $ U_i \cap S \neq \emptyset $.
\item $ |S| \leq 3 x n / q = O (a n \ln{(r n)} / q) $.
\end{enumerate}
\end{lemma}

A well-known property of spanners is \emph{decomposability}.
We will exploit this property to run our dynamic algorithm on carefully chosen subgraphs.
\begin{lemma}[Spanner Decomposability, \cite{BaswanaKS12}]\label{lem:decomposability}
Let $ G = (V, E) $ be an undirected (possibly weighted) graph, let $ E_1, \dots, E_t $ be a partition of the set of edges $ E $, and let, for every $ 1 \leq i \leq t $, $ H_i $ be an $ \alpha $-spanner of $ G_i = (V, E_i) $ for some $ \alpha \geq 1 $.
Then $ H = \bigcup_{i=1}^t H_i $ is an $ \alpha $-spanner of $ G $.
\end{lemma}

In our algorithms we use a reduction for getting a fully dynamic spanner algorithm for an arbitrarily long sequence of updates from a fully dynamic spanner algorithm that only works for a polynomially bounded number of updates.
This is particularly useful for randomized algorithms whose high-probability guarantees are obtained by taking a union bound over a polynomially bounded number of events.
\begin{lemma}[Update Extension, Implicit in~\cite{AbrahamDKKP16}]\label{lem:extending spanner to long sequence}
Assume there is a fully dynamic algorithm for maintaining an $ \alpha $-spanner (for some $ \alpha \geq 1 $) of size at most $ S (m, n, W) $ with worst-case update time $ T (m, n, W) $ for up to $ 4 n^2 $ updates in $ G $.
Then there also is a fully dynamic algorithm for maintaining an $ \alpha $-spanner of size at most $ O (S (m, n, W)) $ with worst-case update time $ O (T (m, n, W)) $ for an arbitrary number of updates.
\end{lemma}

For completeness, we give the proof of this lemma in an appendix.
We remark that is is entirely identical to the one given in~\cite{AbrahamDKKP16}.

\section{Algorithms for Partial Spanner Computation}

Our goal in this section is to describe fully dynamic algorithm for \emph{partial} spanner computation.  We prove lemmas that can informally be summarized as follows: given a graph $G$ with an orientation $\vec{G}$, one can build a very sparse spanner that only covers the edges leaving nodes with large out-degree in $\vec{G}$.  There is a smooth tradeoff between the sparsity of the spanner and the out-degree threshold beyond which edges are spanned.

As a crucial subroutine, our algorithms employ a fully dynamic algorithm for maintaining certain structural information related to a \emph{clustering} of $G$.  We will describe this subroutine first.

\subsection{Maintaining a clustering structure}\label{sec:maintaining clustering}

In the spanner literature, a \emph{clustering} of a graph $G = (V, E)$ is a partition of the nodes $V$ into \emph{clusters} $C_1, \dots, C_k$, as well as a ``leftover'' set of \emph{free} nodes $F$, with the following properties:
\begin{itemize}
\item For each cluster $C_i$, there exists a ``center'' node $x_i \in V$ such that all nodes in $C_i$ are adjacent to $x_i$.
\item The free nodes $F$ are precisely the nodes that are not adjacent to any cluster center.
\end{itemize}

In this paper, we will represent clusterings with a vector $c$ indexed by $V$, such that for any clustered $v \in V$ we have $c[v]$ equal to its cluster center, and for any free $v \in V$ we use the convention $c[v] = \infty$.

We will use the following subroutine in our main algorithms:
\begin{lemma}\label{lem:maintaining clusters worst case}
Given an oriented graph $ \vec{G} = (V, \vec{E}) $ and a set of cluster centers $ S = \{ s_1, \ldots, s_k \} $, there is a fully dynamic algorithm that simultaneously maintains:
\begin{enumerate}
	\item A clustering $c$ of $ G = (V, E) $ with centers $S$
	\item For each node $ v $ and each cluster index $ i \in \{ 1, \ldots, k \} $, the set
	\begin{equation*}
	\inc (v, i) := \{ u \in \inc (v) \mid c [u] = i \}
	\end{equation*}
	(i.e. the incoming neighbors to $v$ from cluster $i$)
	\item For every pair of cluster indices $ i, j \in \{ 1, \ldots, k \} $, the set
	\begin{equation*}
	\inc (i, j) := \{ (u, v) \in \vec{E} \mid c [u] = j, c [v] = i \}
	\end{equation*}
	(i.e. the incoming neighbors to cluster $i$ from cluster $j$).
\end{enumerate}
This algorithm has worst-case update time $ O (\Delta^+ (\vec{G}) \log{n}) $, where $ \Delta^+ (\vec{G}) $ is the maximum out-degree of $ \vec{G} $.
\end{lemma}

The second $\inc(v, i)$ sets will be useful for the $3$-spanner, while the third $\inc(i, j)$ sets will be useful for the $5$-spanner.

The implementation of this lemma is extremely straightforward; it is not hard to show that the necessary data structures can be maintained in the naive way by simply passing a message along the outgoing edges from $u$ and $v$ whenever an edge $(u, v)$ is inserted or deleted.  Due to space constraints, we defer full implementation details and pseudocode to Appendix~\ref{apx:updating clustering algorithm}.

\subsection{Maintaining a partial $3$-spanner}

We next show how to convert Lemma~\ref{lem:maintaining clusters worst case} into a fully dynamic algorithm for maintaining a partial $3$-spanner of a graph, as described in the introduction.  Specifically:

\begin{lemma}\label{lem:3 spanner worst-case fine-grained}
For every integer $ 1 \leq \outputdegree \leq n $, there is a fully dynamic algorithm that takes an oriented graph $\vec{G} = (V, \vec{E})$ on input and maintains subgraphs $A = (V, E_A), \vec{B} = (V, \vec{E}_B)$ (i.e. $\vec{B}$ is oriented but $A$ is not) over a sequence of $4n^2$ updates with the following properties:
\begin{itemize}
\item $ \dist_A (u, v) \leq 3 $ for every edge $ \{u, v\} $ in $E \setminus E_B$
\item $ A $ has size $ | A | = O (n^2 (\log{n}) / \outputdegree + n) $
\item The maximum out-degree of $ \vec{B} $ is $ \Delta^+ (\vec{B}) \leq \outputdegree $.
\item With every update in $ G $, at most $ 4 $ edges are changed in $ \vec{B} $.
\end{itemize}
Further, this algorithm has worst-case update time $  O (\Delta^+ (\vec{G}) \log{n}) $.  The algorithm is randomized, and all of the above properties hold with high probability against an oblivious adversary.
\end{lemma}

Informally, this lemma states the following.  Edges leaving nodes with high out-degree are easy for us to span; we maintain $A$ as a sparse spanner of these edges.  Edges leaving nodes with low out-degree are harder for us to span, and we maintain $\vec{B}$ as a collection of these edges.  

Note that this lemma is considerably \emph{stronger} than the existence of a $3$-spanner.
In particular, by setting $ \outputdegree = \sqrt{n \log{n}} $ and then using $A \cup \vec{B}$ as a spanner of $G$, we obtain a fully dynamic algorithm for maintaining a $3$-spanner:
\begin{corollary}\label{cor:3 spanner worst-case}
There is a fully dynamic algorithm for maintaining a $3$-spanner of size $ O (n^{1 + 1/2} \sqrt{\log{n}}) $ for an oriented graph $ \vec{G} $ with worst-case update time $ O (\Delta^+ (\vec{G}) \log{n}) $. 
The stretch and the size guarantee both hold with high probability against an oblivious adversary.
\end{corollary}
The proof is essentially immediate from Lemma~\ref{lem:3 spanner worst-case fine-grained}; we omit it because it is non-essential.  The detail of handling only $4n^2$ updates is not necessary in this corollary, due to Lemma~\ref{lem:extending spanner to long sequence}.

Looking forward, we will wait until Lemma~\ref{lem:bootstrapping worst case} to show precisely how the extra generality in Lemma~\ref{lem:3 spanner worst-case fine-grained} is useful towards strong worst-case update time.
The rest of this subsection is devoted to the proof of Lemma~\ref{lem:3 spanner worst-case fine-grained}.

\subsubsection{Algorithm}

It will be useful in this algorithm to fix an arbitrary ordering of the nodes in the graph.  This allows us to discuss the ``smallest'' or ``largest'' node in a list, etc.

We initialize the algorithm by determining a set of cluster centers $ S $ via random sampling.
Specifically, every node of $ G $ is added to $ S $ independently with probability $ p = \min (x / \outputdegree, 1) $ where $ x = a \ln{(4 n^5)} + 1 $ for some error parameter $ a \geq 1 $.
We then use the algorithm of Lemma~\ref{lem:maintaining clusters worst case} above to maintain a clustering with $ S = \{s_1, \dots, s_k \} $ as the set of cluster centers.
The subgraphs $ A $ and $ \vec{B} $ are defined according to the following three rules:

\begin{enumerate}
	\item For every clustered node $ v $ (i.e. $ c [v] \neq \infty $), $ A $ contains the edge $ \{ v, c[v] \} $ from $v$ to its cluster center in $S$.\label{edges to cluster centers}
	\item For every clustered node $ v $ (i.e. $ c [v] \neq \infty $) and every cluster index $ 1 \leq i \leq k $, $ A $ contains the edge $ \{ u, v \} $ to the first node $ u \in \inc (v, i) $ (unless $ \inc (v, i) = \emptyset $).\label{edges to clusters}
	
	\item For every node $ u $ and every node $ v $ among the \emph{first} $ \outputdegree $ neighbors of $ u $ in $ N (u) $ (with respect to an arbitrary fixed ordering of the nodes), $\vec{B}$ contains the edge $ (u, v) $.  Alternately, if $ | N(u) | \leq \outputdegree $, then $\vec{B}$ contains all such edges $ (u, v) $.\label{edges to first outgoing neighbors}
\end{enumerate}

We maintain the subgraph $ \vec{B} $ in the following straightforward way.
For every node $ u $ we store $ N (u) $, the set of neighbors of $ u $, in two self-balancing binary search trees: $ N_{\leq \outputdegree} (u) $ for the first $ \outputdegree $ neighbors and $ N_{> \outputdegree} (u) $ for the remaining neighbors.
Every time an edge $ (u, v) $ or an edge $ (v, u) $ is inserted into $ \vec{G} $, we add $ v $ to $ N_{\leq \outputdegree} (u) $ and we add $ (u, v) $ to $ \vec{B} $.
If $ N_{\leq \outputdegree} (u) $ now contains more than $ \outputdegree $ nodes, we remove the largest element $ v' $, add it to $ N_{> \outputdegree} (u) $, and remove $ (u, v') $ from $ \vec{B} $.\footnote{Note that the node $ v' $ that is removed from $ N_{\leq \outputdegree} (u) $ might be the node $ v $ we have added in the first place.}
Similarly, every time an edge $ (u, v) $ or an edge $ (v, u) $ is deleted from $ \vec{G} $, we first check if $ v $ is contained in $ N_{> \outputdegree} (u) $ and if so remove it from $ N_{> \outputdegree} (u) $.
Otherwise, we first remove $ v $ from $ N_{\leq \outputdegree} (u) $ and $ (u, v) $ from $ \vec{B} $.
Then we find the smallest node $ v' $ in $ N_{> \outputdegree} (u) $, remove $v'$ from $ N_{> \outputdegree} (u) $, add $ v' $ to $ N_{\leq \outputdegree} (u) $, and add $ (u, v) $ to $ \vec{B} $.

We now explain how to maintain the subgraph $ A $.  As an underlying subroutine, we use the algorithm of Lemma~\ref{lem:maintaining clusters worst case} to maintain a clustering w.r.t. centers $ S $.
On each edge insertion/deletion, we first update the clustering, and then perform the following steps:
\begin{enumerate}
	\item For every node $v$ for which $ c [v] $ has just changed from some center $ s_i $ to some other center $ s_j $, we remove the edge $ \{v, s_i\} $ from $ A $ (if $ i \neq \infty $) and add the edge $ \{v, s_j\} $ to $ A $ (if $ j \neq \infty $).\label{node changes cluster}
	\item For every node $ u $ that has been added to $ \inc (v, i) $ for some node $ v $ and some $ 1 \leq i \leq k $, we check if $ u $ is now the first node in $ \inc (v, i) $.  If so, we add the edge $ \{u, v\} $ to $ A $ and remove the edge $ \{u', v\} $ for the previous first node $ u' $ of $ \inc (v, i) $ (if $ \inc (v, i) $ was previously non-empty).\label{node added to list}
	\item For every node $ u $ that is removed from $ \inc (v, i) $ for some node $ v $ and some $ 1 \leq i \leq k $, we check if $ u $ was the first node in $ \inc (v, i) $ and if so remove the edge $ \{u, v\} $  from $ A $ and add the edge $ \{u', v\} $ for the new first node $ u' $ of $ \inc (v, i) $ (if $ \inc (v, i) $ is still non-empty).\label{node removed from list}
\end{enumerate}

\subsubsection{Analysis}

To bound the update time required by this algorithm, we will argue that we spend $ O (\Delta^+ (\vec{G}) \log{n}) $ time per update maintaining $A$, and $O(\log n)$ time per update maintaining $\vec{B}$ (which, in our applications, is always dominated by $ O (\Delta^+ (\vec{G}) \log{n}) $).
By Lemma~\ref{lem:maintaining clusters worst case}, the clustering structure can be updated in time $ O (\Delta^+ (\vec{G}) \log{n}) $.
Each operation in steps~\ref{node changes cluster},~\ref{node added to list}, and~\ref{node removed from list} above can be charged to the corresponding changes in $ s_i $ and $ \inc (v, i) $ and thus can also be carried out within the same $ O (\Delta^+ (\vec{G}) \log{n}) $ time bound.
Updating the subgraph $ \vec{B} $ takes time $ O (\log{n}) $, since we must perform a constant number of queries and updates in the corresponding self-balancing binary search trees.

We now show that the subgraphs $ A $ and $ \vec{B} $ have all of the properties claimed in Lemma~\ref{lem:3 spanner worst-case fine-grained}.
First, we will discuss the sparsity bounds on $A$ and $\vec{B}$.
Observe that rule~\ref{edges to cluster centers} contributes at most $n$ edges to $A$, since every node is contained in at most one cluster.
Next, recall that the number of cluster centers $ S $ is $|S| = k = O (n (\log{n}) / d) $ (by Lemma~\ref{lem:random hitting set}, with high probability).
Thus, $ A $ contains only $ O(n k) = O (n^2 (\log{n}) / d) $ edges due to rule~\ref{edges to clusters}.
As the only edges of $ \vec{B} $ come from rule~\ref{edges to first outgoing neighbors}, the maximum out-degree in $ \vec{B} $ is $ \outputdegree $.
The claimed sparsity bounds therefore hold.
Furthermore, with every insertion or deletion of an edge $ \{u, v\} $ in $ G $, at most one edge is added to or removed from the first $ \outputdegree $ neighbors of $ u $ and $ v $, respectively.
This implies that there are at most $4$ changes to $ \vec{B} $ with every update in $ G $.
It now only remains to show that $ A $ is a $3$-spanner of $ G \setminus B $.

\begin{lemma}
For up to $ 4 n^3 $ updates, $ \dist_A (u, v) \leq 3 $ for every edge $ \{u, v\} $ in $E \setminus E_B$ with high probability.
\end{lemma}

\begin{proof}
Let $ \{ u, v \} $ be an edge of $E \setminus E_B$.
Assume  without loss of generality that the edge is oriented from $ u $ to $ v $ in $ \vec{G} $.
As $ \{ u, v \} $ is not contained in $ B $, by rule~\ref{edges to first outgoing neighbors} above we have $ | N (u) | > \outputdegree $.
Thus, by Lemma~\ref{lem:random hitting set}, since the cluster centers $S$ were chosen by random sampling, with high probability there exists a cluster center in the first $ \outputdegree $ outgoing neighbors of each node in all of up to $ 4 n^3 $ different versions of $G$ (i.e. one version for each of the $4n^3$ updates considered).
Therefore $ c[u] = i $ for some $ 1 \leq i \leq k $ and, by rule~\ref{edges to cluster centers}, $ A $ contains the edge $ \{ u, s_i \} $.
Since $ c[u] = i $, and $ u $ is an incoming neighbor of $ v $ in $ \vec{G} $, we have $ \inc (v, i) \neq \emptyset $, and thus, for the first element $ u' $ of $ \inc (v, i) $, $ A $ contains the edge $ \{ u', v \} $ (by rule~\ref{edges to clusters}).
As $ c [u'] = i $, $ A $ contains the edge $ \{ s_i, u' \} $ by rule~\ref{edges to cluster centers}.
This means that $ A $ contains the edges $ \{ u, s_i \} $, $ \{ s_i, u' \} $, and $ \{ u', v \} $, and thus there is a path from $ u $ to $ v $ of length~$ 3 $ in $ A $ as desired.
\end{proof}
This now also completes the proof of Lemma~\ref{lem:3 spanner worst-case fine-grained}.

\subsection{$5$-spanner}
 
The $5$-spanner algorithm is very similar to the $3$-spanner algorithm above, but we define the edges of the spanner in a slightly different way.
Instead of including an edge from each node to each cluster, we have an edge between each \emph{pair} of clusters.
Thus, the subgraphs $ A $ and $ \vec{B} $ are defined according to the following three rules:
\begin{enumerate}
	\item For every clustered node $ v $ (i.e. $ c [v] \neq \infty $), $ A $ contains the edge $ \{ v, c[v] \} $ from $v$ to its cluster center in $S$.\label{edges to cluster centers 5-spanner}
	\item For every pair of distinct cluster indices $ 1 \leq i, j \leq k $, $ A $ contains the edge $ \{ u, v \} $, where $ \{ u, v \} $ is the first element in $ \inc (i, j) $ (unless $ \inc (i, j) = \emptyset $).\label{edges between clusters 5-spanner}
	\item For every node $ u $ and every node $ v $ among the \emph{first} $ \outputdegree $ neighbors of $ u $ in $ N (u) $ (with respect to an arbitrary fixed ordering of the nodes), $\vec{B}$ contains the edge $ (u, v) $.  Alternately, if $ | N(u) | \leq \outputdegree $, then $\vec{B}$ contains all such edges $ (u, v) $..\label{edges to first outgoing neighbors 5-spanner}
\end{enumerate}

Beyond this slightly altered definition, we use the same approach for maintaining $ A $ and $ \vec{B} $ as in the $3$-spanner.
The guarantee on the stretch can be proved as follows.
\begin{lemma}
For up to $ 4 n^3 $ updates, $ \dist_A (u, v) \leq 5 $ for every edge $ \{u, v\} $ in $E \setminus E_B$ with high probability.
\end{lemma}

\begin{proof}
Let $ \{ u, v \} $ be an edge of $E \setminus E_B$.
Assume  without loss of generality that the edge is oriented from $ u $ to $ v $ in $ \vec{G} $.
As $ \{ u, v \} $ is not contained in $ B $, by rule~\ref{edges to first outgoing neighbors 5-spanner} above we have $ | N (v) | > \outputdegree $.
We now apply Lemma~\ref{lem:random hitting set} to argue that there is a cluster center in the first $ \outputdegree $ outgoing neighbors of each node in up to $ 4 n^3 $ versions of the graph (one version for each update to be considered).
Thus, $ N (v) $ contains a cluster center from $ S $ with high probability.
Therefore $ c[v] = i $ for some $ 1 \leq i \leq k $ and, by rule~\ref{edges to cluster centers}, $ A $ contains the edge $ \{ v, s_i \} $.
By the same argument, $ N (u) $ contains a cluster center from $ S $ with high probability and thus $ A $ contains an edge $ \{ u, s_j \} $ where $ c[u] = j $ for some $ 1 \leq j \leq k $.
Since $ c[v] = i $, $ c[u] = j $, and $ u $ is an incoming neighbor of $ v $ in $ \vec{G} $, we have $ \inc (i, j) \neq \emptyset $, and thus, for the first element $ (u', v') $ of $ \inc (i, j) $, $ A $ contains the edge $ \{ u', v' \} $ (by rule~\ref{edges between clusters 5-spanner}).
As $ c[v'] = i $, $ c[u'] = j $, $ A $ contains the edges $ \{ v', s_i \} $ and $ \{ u', s_j \} $ by rule~\ref{edges to cluster centers 5-spanner}.
This means that $ A $ contains the edges $ \{ u, s_i \} $, $ \{ s_i, u' \} $, $ \{ u', v' \} $, $ \{ v', s_i \} $ and $ \{ s_i, v \} $, and thus there is a path from $ u $ to $ v $ of length~$ 5 $ in $ A $ as desired.
\end{proof}
Note that in this proof we exploit the fact that we have cluster centers for both $ u $ and $ v $ whenever the edge $ \{u, v\} $ is missing.
This motivates our design choice for considering the whole neighborhood of a node to determine its cluster.
If we only considered cluster centers in the outgoing neighbors of a node, the resulting clustering would still be good enough for the $3$-spanner, but the argument above for the $5$-spanner would break down.

All other properties of the $5$-spanner can be proved in an essentially identical manner to the $3$-spanner.
We can summarize the obtained guarantees as follows.
\begin{lemma}\label{lem:5 spanner worst-case fine-grained}
For every integer $ 1 \leq \outputdegree \leq n $, there is a fully dynamic algorithm that takes an oriented graph $\vec{G} = (V, \vec{E})$ on input and maintains subgraphs $A = (V, E_A), \vec{B} = (V, \vec{E}_B)$ (i.e. $\vec{B}$ is oriented but $A$ is not) over a sequence of $4n^2$ updates with the following properties:
\begin{itemize}
\item $ \dist_A (u, v) \leq 5 $ for every edge $ \{u, v\} $ in $E \setminus E_B$
\item $ A $ has size $ | A | = O ((n^2 \log^2{n}) / \outputdegree^2 + n) $
\item The maximum out-degree of $ \vec{B} $ is $ \Delta^+ (\vec{B}) \leq \outputdegree $.
\item With every update in $ G $, at most $ 4 $ edges are changed in $ \vec{B} $.
\end{itemize}
Further, this algorithm has worst-case update time $ O (\Delta^+ (\vec{G}) \log{n}) $.  The algorithm is randomized, and all of the above properties hold with high probability against an oblivious adversary.
\end{lemma}

Once again, this lemma generalizes the construction of a sparse $5$-spanner.
By setting $ \outputdegree = (n \log{n})^{2/3} $ we can obtain:
\begin{corollary}\label{cor:5 spanner worst-case}
There is a fully dynamic algorithm for maintaining a $5$-spanner of size $ O (n^{1 + 1/3} \log^{2/3}{n}) $ for an oriented graph $ \vec{G} $ with worst-case update time $ O (\Delta^+ (\vec{G}) \log{n}) $. 
The stretch and the size guarantee both hold with high probability against an oblivious adversary.
\end{corollary}

\section{Out-degree Reduction for Improved Update Time}\label{lem:bootstrapping worst case}

Our goal is now to use Lemmas \ref{lem:3 spanner worst-case fine-grained} and \ref{lem:5 spanner worst-case fine-grained} to obtain spanner algorithms with sublinear update time.  Since we obtain our $3$-spanner and $5$-spanner in an essentially identical manner, we will explain only the $3$-spanner in full detail, and then sketch the $5$-spanner construction.

We next establish the following simple generalization of Lemma \ref{lem:3 spanner worst-case fine-grained}:

\begin{lemma}\label{lem:degree reduction 3-spanner}
For every integer $ 1 \leq \inputdegree \leq n $ and $ 1 \leq \outputdegree \leq n $, there is a fully dynamic algorithm that takes an oriented graph $\vec{G} = (V, \vec{E})$ on input and maintains subgraphs $A = (V, E_A), \vec{B} = (V, \vec{E}_B)$ (i.e. $\vec{B}$ is oriented but $A$ is not) over a sequence of $4n^2$ updates with the following properties:
\begin{itemize}
\item $ \dist_A (u, v) \leq 3 $ for every edge $ \{u, v\} $ in $E \setminus E_B$
\item $ A $ has size $ | A | = O (\Delta^+ (\vec{G}) n^2 (\log n) / (\inputdegree \outputdegree)) $
\item The maximum out-degree of $ \vec{B} $ is $ \Delta^+ (\vec{B}) \leq \Delta^+ (\vec{G}) \cdot \outputdegree / \inputdegree $.
\item With every update in $ G $, at most $ 4 $ edges are changed in $ \vec{B} $.
\end{itemize}
Further, this algorithm has worst-case update time $ O (\inputdegree \log{n}) $. The algorithm is randomized, and all of the above properties hold with high probability against an oblivious adversary.
\end{lemma}
In particular, Lemma~\ref{lem:3 spanner worst-case fine-grained} is the special case of this lemma in which $\inputdegree = \Delta^+ (\vec{G})$.

\begin{proof}
We orient each incoming edge of $ G $ in an arbitrary way.
We then maintain a partitioning of the (oriented) edges of $ G $ into $ t := \lceil \Delta^+ (\vec{G}) / \inputdegree \rceil $ groups, such that in each group each node has at most $ \inputdegree $ outgoing edges.
Specifically, we perform this partitioning by maintaining the current out-degree of each node $u$ in $\vec{G}$, and we assign a new edge $(u, v)$ which is the $x^{th}$ edge leaving $u$ in $\vec{G}$ to the subgraph $\vec{G}_{\lceil x/s \rceil}$.
In this way, we form $ t $ subgraphs $ \vec{G}_1, \ldots \vec{G}_t $ of $ \vec{G} $, each of which has $\Delta^+(\vec{G}_i) \le s$.

We now run the algorithm of Lemma~\ref{lem:3 spanner worst-case fine-grained} on each $ \vec{G}_i $ to maintain, for each $ 1 \leq i \leq t $, two subgraphs $ A_i $ and $ \vec{B}_i $ as specified in the lemma.
Let $ A = \bigcup A_i $ and $ \vec{B} = \bigcup \vec{B_i} $ denote the unions of these subgraphs.

Observe that every update in $ G $ only changes exactly one of the subgraphs $ \vec{G}_i $ and thus only must be executed in one corresponding instance of the algorithm of Lemma~\ref{lem:3 spanner worst-case fine-grained}.
As we have ``artificially'' bounded the maximum out-degree of every subgraph $ \vec{G}_i $ by $ \inputdegree $, the claimed bounds on the update time and the properties of $A$ and $\vec{B}$ now follow simply from Lemma~\ref{lem:3 spanner worst-case fine-grained}.
\end{proof}

We now recursively apply the ``out-degree reduction'' of the previous lemma to obtain subgraphs $ \vec{B} $ of smaller and smaller out-degree.  Finally, at bottom level, the maximum out-degree is small enough that we can apply a ``regular'' spanner algorithm to it.

\begin{theorem} \label{thm:3 span}
There is a fully dynamic algorithm for maintaining a $3$-spanner of size $ O(n^{1+1/2} \log^{1/2}{n} \log{\log{n}}) $ with worst-case update time $ O (n^{3/4} \log^{4}{n}) $.
\end{theorem}

\begin{proof}
Our spanner construction is as follows (we temporarily omit details related to parameter choices, which influence the resulting update time).
Apply Lemma \ref{lem:degree reduction 3-spanner} to obtain subgraphs $A_1, \vec{B}_1$.
Include all edges in $A$ in the spanner, and then recursively apply Lemma \ref{lem:degree reduction 3-spanner} to $\vec{B}$ to obtain $A_2, \vec{B}_2$.
Repeat to depth $ \ell $ (for some parameter $\ell$ that will be chosen later).
At bottom level, instead of recursing, we apply the algorithm from Corollary~\ref{cor:3 spanner worst-case} to obtain a $3$-spanner of $\vec{B}$.

More formally, we set $ \vec{B}_0 = \vec{G}_0 $, and for every $ 1 \leq j \leq \ell $ we let $ A_j $ and $ \vec{B}_j $ be the graphs maintained by the algorithm of Lemma~\ref{lem:degree reduction 3-spanner} on input $ \vec{B}_{j-1} $ using parameters $ \inputdegree $ and $ \outputdegree_j $ to be chosen later.\footnote{Note that the parameter $ s $ is the same for all levels of the recursion, whereas the parameter $ \outputdegree_j $ is not.}
Further, we let $ H' $ be the spanner maintained by the algorithm of Corollary~\ref{cor:3 spanner worst-case} on input $ \vec{B}_\ell $.
The resulting graph maintained by our algorithm is $ H = \bigcup_{1 \leq j \leq \ell} A_j \cup H' $.
Then, by Lemma \ref{lem:degree reduction 3-spanner}, we have the following properties for every $ 1 \leq j \leq \ell $:
\begin{itemize}
\item $ \dist_{A_j} (u, v) \leq 3 $ for every edge $ \{u, v\} $ in $ B_{j-1} \setminus B_j $
\item $ A_j $ has size $ | A_j | = O (\Delta^+ (\vec{B}_{j-1}) n^2 (\log n) / (\inputdegree \outputdegree_j)) $
\item The maximum out-degree of $ \vec{B_j} $ is $ \Delta^+ (\vec{B}_j) \leq \Delta^+ (\vec{B}_{j-1}) \cdot \outputdegree_j / \inputdegree $.
\item With every update in $ \vec{B}_{j-1} $, at most $ 4 $ edges are changed in $ \vec{B_j} $.
\end{itemize}

It is straightforward to see that the resulting graph $ H $ is a $3$-spanner of $G$:
At each level $ j $ of the recursion, $ A_j $ spans all edges of $B_{j-1}$  \emph{except} those that appear in the current subgraph $\vec{B}_j$.
Thus, at bottom level, the only non-spanned edges of $G$ are those in the final subgraph $\vec{B}_\ell$.
For these edges we explicitly add a $3$-spanner $H'$ of $\vec{B}_\ell$ to $ H $.
By Lemma \ref{lem:span adjacency}, this suffices to produce a $3$-spanner of all of $G$.

Now that we have correctness of the construction, it remains to bound the number of edges in the output spanner.
First, observe that, by induction,
\begin{align*}
\Delta^+ (\vec{B}_{j}) \leq n \cdot \prod_{1 \leq j' \leq j} \outputdegree_{j'} / \inputdegree^j
\end{align*}
for all $ 1 \leq j \leq \ell $.
Since additionally $H'$ has size $O(n^{1 + 1/2} \log^{1/2} n)$ by Corollary~\ref{cor:3 spanner worst-case}, the total number of edges in $ H $ is
\begin{align*}
|H| = \sum_{1 \leq j \leq \ell} | A_j | + | H' | &\leq \sum_{1 \leq j \leq \ell} O \left( \frac{\Delta^+ (B_{j-1}) n^2 \log n}{\inputdegree \outputdegree_j} \right) + O(n^{1 + 1/2} \log^{1/2} n) \\
 &\leq \sum_{1 \leq j \leq \ell} O \left( \frac{ \left( \prod \limits_{1 \leq j' \leq j-1} \outputdegree_{j'} \right) n^3 \log n}{\inputdegree^j \outputdegree_j} \right) + O(n^{1 + 1/2} \log^{1/2} n) \, .
\end{align*}
Thus, our spanner satisfies the claimed sparsity bound so long as the union of all $\ell$ of the $A_j$ subgraphs fit within the claimed sparsity bound; this will be the case if we balance all summands.

We next bound the update time of our algorithm.
Each change to some $\vec{B}_j$ causes at most $4$ changes in the next level $\vec{B}_{j+1}$, and thus the number of changes to $ \vec{B}_j $ can propagate exponentially.
Thus, for every $ 0 \leq j \leq \ell-1 $, a single update in $ \vec{G} $ could cause at most $ 4^j $ changes to $\vec{B}_j$.
Each of the $ \ell $ instances of the algorithm of Lemma~\ref{lem:degree reduction 3-spanner} has a worst-case update time of $ O (\inputdegree \log n) $ and the algorithm of Corollary~\ref{cor:3 spanner worst-case} has a worst-case update time of $ \Delta^+ (\vec{B}_\ell \log n) $.
Since
\begin{align*}
\Delta^+ (\vec{B}_\ell) \leq n \cdot \prod_{1 \leq j \leq \ell} \outputdegree_j / \inputdegree^\ell
\end{align*}
the worst-case update time of our overall algorithm is
\begin{equation*}
O \left( \left(\sum_{j=0}^{\ell-1} 4^j \inputdegree + 4^\ell \Delta^+ (\vec{B}_\ell) \right) \cdot \log n \right) \le O \left( \left( s + \frac{n \cdot \prod \limits_{1 \leq j \leq \ell} \outputdegree_j}{\inputdegree^\ell} \right) \cdot 4^\ell \log n \right) \, .
\end{equation*}

Our goal is now to choose parameters $\inputdegree_j, \outputdegree, \ell$ to minimize this expression subject to the constraint on spanner size given above.
To achieve this, we set parameters as follows:
\begin{align*}
\ell &= \log{\log{n}} \, ,\\
\inputdegree	&= n^{(3 \cdot 2^\ell - 1) / (2^{\ell+2} - 2)}		\log{n}	\, , \text{ and} \\
\outputdegree_j	&= n^{(3 \cdot 2^\ell - 2^{j-1} - 1) / (2^{\ell+2} - 2)}	\log{n}	\, .
\end{align*}
These values were obtained with the help of a computer algebra solver, so we do not have explicit computations to show for them.
\end{proof}

We now turn our attention to the $5$-spanner.
Similar to Lemma \ref{lem:degree reduction 3-spanner} above, we can use Lemma~\ref{lem:5 spanner worst-case fine-grained} to perform a similar out-degree reduction step for our dynamic 5-spanner algorithm.

\begin{lemma}\label{lem:degree reduction 5-spanner}
For every integer $ 1 \leq \inputdegree \leq n $ and $ 1 \leq \outputdegree \leq n $, there is a fully dynamic algorithm that takes an oriented graph $\vec{G} = (V, \vec{E})$ on input and maintains subgraphs $A = (V, E_A), \vec{B} = (V, \vec{E}_B)$ (i.e. $\vec{B}$ is oriented but $A$ is not) over a sequence of $4n^2$ updates with the following properties:
\begin{itemize}
\item $ \dist_A (u, v) \leq 5 $ for every edge $ \{u, v\} $ in $E \setminus E_B$
\item $ A $ has size $ | A | = O (\Delta^+ (\vec{G}) n^2 (\log^2 n) / (\inputdegree \outputdegree^2)) $
\item The maximum out-degree of $ \vec{B} $ is $ \Delta^+ (\vec{B}) \leq \Delta^+ (\vec{G}) \cdot \outputdegree / \inputdegree $.
\item With every update in $ G $, at most $ 4 $ edges are changed in $ \vec{B} $.
\end{itemize}
Further, this algorithm has worst-case update time $ O (\inputdegree \log{n}) $. The algorithm is randomized, and all of the above properties hold with high probability against an oblivious adversary.
\end{lemma}
The proof of this lemma is essentially identical to the proof of Lemma \ref{lem:degree reduction 3-spanner} and has thus been omitted.

Just as in the case of the 3-spanner, we use this lemma to show:
\begin{theorem} \label{thm:5 span}
There is a fully dynamic algorithm for maintaining a $5$-spanner of size $ O (n^{1+1/3} \log^{2/3}{n} \log{\log{n}}) $ with worst-case update time $ O (n^{5/9} \log^{4}{n}) $.
\end{theorem}
\begin{proof}
The proof is identical to the proof of Theorem \ref{thm:3 span}, except that the proper parameter balance is now:
\begin{align*}
\ell &= \log{\log{n}} \, ,\\
\inputdegree	&= n^{(5 \cdot 3^\ell - 2^{\ell+1}) / (3^{\ell+2} - 3 \cdot 2^{\ell+1})}				\log{n}	\, , \text{ and} \\
\outputdegree_j	&= n^{(5 \cdot 3^\ell - 3^{j-1} 2^{\ell-j+2} - 2^{\ell+1}) / (3^{\ell+2} - 3 \cdot 2^{\ell+1})}	\log{n}	\, .
\end{align*}
\end{proof}

Finally, we can also show:
\begin{theorem}
There is a fully dynamic algorithm for maintaining a $5$-spanner of size $ O (n^{1+1/2} \log^{1/2}{n} \log{\log{n}}) $ with worst-case update time $ O (n^{1/2} \log^{4}{n}) $.
\end{theorem}

\begin{proof}
The proof is identical to the proof of Theorems~\ref{thm:3 span} and~\ref{thm:5 span}, except that we now use the parameter balance
\begin{align*}
\ell &= \log{\log{n}} \, ,\\
\inputdegree	&= n^{(3^{\ell+1} - 2^\ell) / (2 \cdot 3^{\ell+1} - 2^{\ell+2})}				\log{n}	\, , \text{ and} \\
\outputdegree_j	&= n^{(3^{\ell+1} - 3^j \cdot 2^{\ell-j} - 2^\ell) / (2 \cdot 3^{\ell+1} - 2^{\ell+2})}	\log{n}	\, .
\end{align*}
and we maintain the dynamic $3$-spanner $ H' $ of size $ O (n^{1 + 1/2} \log^{1/2}{n}) $ from Corollary~\ref{cor:5 spanner worst-case} at bottom level.
\end{proof}
This spanner has non-optimal size/stretch tradeoff, but enjoys the best worst-case update time that we are currently able to construct.

\subsection*{Acknowledgements}

We want to thank Seeun William Umboh for many fruitful discussions at Simons.

\printbibliography[heading=bibintoc] 

\newpage
\appendix

\section{Updating the Clustering}\label{apx:updating clustering algorithm}

In the following we give the straightforward algorithm for maintaining the clustering with worst-case update time proportional to the maximum out-degree of the original graph mentioned in Section~\ref{sec:maintaining clustering}.
To make the presentation of this algorithm more succinct we assume that there is some (arbitrary, but fixed) ordering on the nodes.
Furthermore, we assume that the nodes of $ S $ are given according to this order, i.e. $ s_1 \leq s_2 \leq \dots \leq s_k $.
For every node $ v $, we maintain $ c [v] $ as the smallest $ i $ such that $ s_i $ is a neighbor of $ v $ (or $ \infty $ if no such neighbor exists).
Additionally, we naturally extend the sets $ \inc (v, i) $ and $ \inc (i, j) $ to the case $ i, j \in \{1, \dots, k, \infty\} $.

We begin with an empty graph $G = (V, \emptyset)$, a cluster vector $c$ with $c[v] = \infty$ for all $v \in V$, and empty sets $\inc(i, j)$ and $\inc(i, v)$ for all cluster indices $1 \le i, j \le k$ and $v \in V$.  We then modify these data structures under edge insertions and deletions as follows.

Correctness of the algorithms that follow is immediate, and is not shown formally.

\subsection{Insertion of an edge $(u, v)$}
	\begin{itemize}
		\item Add $ u $ to $ \inc (v, i) $ for $ i = c [u] $.
		\item Add $ (u, v) $ to $ \inc (j, i) $ for $ i = c [u] $ and $ j = c [v] $
		\item If $ u = s_i $ for some $ 1 \leq i \leq k $:
			\begin{itemize}
				\item Set $ j = c[v] $ (might be $ \infty $)
				\item Add $ u $ to $ C [v] $.
				\item If $ i < j $:
				\begin{itemize}
					\item Set $ c[v] = i $
					\item For every outgoing neighbor $ v' $ of $ v $:

						Remove $ v $ from $ \inc (v', j) $ and add $ v $ to $ \inc (v', i) $.
						
						Remove $ (v, v') $ from $ \inc (i', j) $ and add $ (v, v') $ to $ \inc (i', i) $ where $ i' = c[v'] $.
				\end{itemize}
			\end{itemize}
		\item If $ v = s_i $ for some $ 1 \leq i \leq k $:
			\begin{itemize}
				\item Set $ j = c[u] $ (might be $ \infty $)
				\item Add $ v $ to $ C [u] $.
				\item If $ i < j $:
				\begin{itemize}
					\item Set $ c[u] = i $
					\item For every outgoing neighbor $ v' $ of $ u $:
					
						Remove $ u $ from $ \inc (v', j) $ and add $ u $ to $ \inc (v', i) $.
						
						Remove $ (u, v') $ from $ \inc (i', j) $ and add $ (u, v') $ to $ \inc (i', i) $ where $ i' = c[v'] $.
				\end{itemize}
			\end{itemize}
	\end{itemize}

\subsection{Deletion of an edge $(u, v)$:}

	\begin{itemize}
		\item Remove $ u $ from $ \inc (v, i) $ for $ i = c [u] $.
		\item Remove $ (u, v) $ from $ \inc (j, i) $ for $ i = c [u] $ and $ j = c [v] $
		\item If $ u = s_i $ for some $ 1 \leq i \leq k $:
			\begin{itemize}
				\item Remove $ u $ from $ C [v] $.
				\item If $ c[v] = s_i $:
					\begin{itemize}
						\item Let $ j $ be minimal such that $ s_j $ is in $ C [v] $ (might be $ \infty $)
						\item Set $ c[v] = j $
						\item For every outgoing neighbor $ v' $ of $ v $:
					
							Remove $ v $ from $ \inc (v', i) $ and add $ v $ to $ \inc (v', j) $.
						
							Remove $ (v, v') $ from $ \inc (i', j) $ and add $ (v, v') $ to $ \inc (i', i) $ where $ i' = c[v'] $
					\end{itemize}
			\end{itemize}
		\item If $ v = s_i $ for some $ 1 \leq i \leq k $:
			\begin{itemize}
				\item Remove $ v $ from $ C [u] $.
				\item If $ c[u] = s_i $:
					\begin{itemize}
						\item Let $ j $ be minimal such that $ s_j $ is in $ C [u] $ (might be $ \infty $)
						\item Set $ c[u] = j $
						\item For every outgoing neighbor $ v' $ of $ u $:
					
							Remove $ u $ from $ \inc (v', i) $ and add $ u $ to $ \inc (v', j) $.
						
							Remove $ (u, v') $ from $ \inc (i', j) $ and add $ (u, v') $ to $ \inc (i', i) $ where $ i' = c[v'] $
					\end{itemize}
			\end{itemize}
	\end{itemize}

\section{Proof of Lemma \ref{lem:extending spanner to long sequence}}

We exploit the decomposability of spanners.
We maintain a partition of $ G $ into two disjoint subgraphs $ G_1 $ and $ G_2 $ and run two instances $ A_1 $ and $ A_2 $ of the dynamic algorithm on $ G_1 $ and $ G_2 $, respectively.
These two algorithms maintain a $ t $-spanner of $ H_1 $ of $ G_1 $ and a $ t $-spanner $ H_2 $ of $ G_2 $.
By Lemma \ref{lem:decomposability}, the union $ H = H_1 \cup H_2 $ is a $ t $-spanner of $ G = G_1 \cup G_2 $.

We divide the sequence of updates into phases of length $ n^2 $ each.
In each phase of updates one of the two instances $ A_1 $, $ A_2 $ is in the state \emph{growing} and the other one is in the state \emph{shrinking}.
$ A_1 $ and $ A_2 $ switch their states at the end of each phase.
In the following we describe the algorithm's actions during one phase.
Assume without loss of generality that, in the phase we are fixing, $ A_1 $ is growing and $ A_2 $ is shrinking.

At the beginning of the phase we restart the growing instance $ A_1 $.
We will orchestrate the algorithm in such a way that at the beginning of the phase $ G_1 $ is the empty graph and $ G_2 = G $.
After every update in $ G $ we execute the following steps:
\begin{enumerate}
	\item If the update was the insertion of some edge $ e $, then $ e $ is added to the graph $ G_1 $ and this insertion is propagated to the \emph{growing} instance $ A_1 $.
	\item If the update was the deletion of some edge $ e $, then $ e $ is removed from the graph $ G_i $ it is contained in and this deletion is propagated to the corresponding instance $ A_i $.
	\item In addition to processing the update in $ G $, if $ G_2 $ is non-empty, then one arbitrary edge~$ e $ is first removed from $ G_2 $ and deleted from instance $ A_2 $ and then added to $ G_1 $ and inserted into instance $ A_1 $.
\end{enumerate}
Observe that these rules indeed guarantee that $ G_1 $ and $ G_2 $ are disjoint and together contain all edges of $ G $.
Furthermore, since the graph $ G_2 $ of the shrinking instance has at most $ n^2 $ edges at the beginning of the phase, the length of $ n^2 $ updates per phase guarantees that $ G_2 $ is empty at the end of the phase.
Thus, the growing instance always starts with an empty graph $ G_1 $.

As both $ H_1 $ and $ H_2 $ have size at most $ S (n, m, W) $, the size of $ H = H_1 \cup H_2 $ is $ O (S (n, m, W)) $.
With every update in $ G $ we perform at most $ 2 $ updates in each of $ A_1 $ and $ A_2 $.
It follows that the worst-case update time of our overall algorithm is $ O (T (m, n, W)) $.
Furthermore since each of the instances $ A_1 $ and $ A_2 $ is restarted every other phase, each instance of the dynamic algorithm sees at most $ 4 n^2 $ updates before it is restarted.

\end{document}